\documentclass[12pt]{article}
\usepackage{amsmath,amsfonts,amsthm,amssymb,bbm}
\usepackage{enumerate} \usepackage{graphicx}
\usepackage{newlfont}
\usepackage{color}

\newcommand{\essinf}{\ensuremath{ \mathrm{ess}\inf }}
\newcommand{\esssup}{\ensuremath{ {\mathrm{ess}\sup }}}
\newtheorem{lemma}{Lemma}[section]
\newtheorem{proposition}[lemma]{Proposition}
\newtheorem{corollary}[lemma]{Corollary}
\newtheorem{definition}[lemma]{Definition}

\theoremstyle{remark}
\newtheorem{remark}[lemma]{Remark}
\newtheorem{example}[lemma]{Example}
\newcommand{\Ft}{\ensuremath{\mathcal{F}_t}}
\newcommand{\Fs}{\ensuremath{\mathcal{F}_s}}
\newcommand{\FT}{\ensuremath{\mathcal{F}_T}}

\newcommand{\Lbbt}{\ensuremath{L^{bb}_t}}
\newcommand{\LbbT}{\ensuremath{L^{bb}_T}}
\newcommand{\St}{\ensuremath{\mathcal{S}_t}}
\newcommand{\T}{\ensuremath{\mathcal{T}}}
\newcommand{\bt}{\ensuremath{\beta_t}}
\newcommand{\bs}{\ensuremath{\beta_s}}

\newcommand{\Lbat}{\ensuremath{L^{ba}_t}}
\newcommand{\rtz}{\ensuremath{\rho^z_t}}
\newcommand{\rsz}{\ensuremath{\rho^z_s}}
\newcommand{\rtzi}{\ensuremath{\rho^{z_i}_t}}
\newcommand{\stz}{\ensuremath{{\sigma}^z_t}}
\newcommand{\stzi}{\ensuremath{{\sigma}^{z_i}_t}}

\begin{document}
\title {Dynamic quasi concave performance measures}
\author {Sara Biagini \thanks{University of Pisa, Italy}  \and  Jocelyne Bion-Nadal \thanks{CNRS, CMAP - Ecole Polytechnique, France}}\maketitle
\begin{abstract}
We define   Conditional  quasi concave Performance Measures (CPMs), on
  random variables bounded from below, to accommodate for additional information.
Our notion  encompasses a wide variety of cases, from conditional expected utility and certainty equivalent to conditional acceptability indexes.  We provide the  characterization of  a  CPM   in terms of an induced  family of conditional convex risk measures.  In the case of  indexes these  risk measures are coherent.  Then,  Dynamic Performance Measures  (DPMs) are introduced and the problem of  time consistency  is addressed. The definition of time consistency chosen here ensures that the positions which are considered good tomorrow are already considered good today.  We prove the equivalence between time consistency for a DPM and  \emph{weak acceptance consistency} for the induced  families of  risk measures. Finally, we extend CPMs and DPMs to dividend processes.


\end{abstract}

\section{Introduction}

Portfolio selection and the companion asset allocation  are undoubtedly  one of the most important problems in Finance and Insurance.  In order to make an optimal choice, some performance criterion must be selected.  Since Markowitz's seminal work   many  criteria have  been proposed,  from expected utility   to ratios which seek a balance between reward and risk, such as the  Sharpe Ratio index.   The Sharpe Ratio of a position consists of the ratio of expected value over standard deviation.  It has many pleasant features: clear meaning, easiness of computation,   scale invariance and has consequently had a great  success in the industry. From a normative point of view, and outside a Gaussian context, it is well-known that this  index  has many drawbacks:  it is not monotone and   it may seriously underestimate arbitrage (see e.g. Bernardo and Ledoit \cite[Section 2.2.1]{bl}). Thus, a  whole class of indexes have been developed  to improve the Sharpe Ratio in investment evaluation, while preserving its good features.   Some examples are the Sortino Ratio,  the Gain-Loss Ratio by Bernardo and Ledoit, or, more generally, the ratio of  reward to lower  $p$-th partial moments in Leitner \cite{lei-lpm}, and  generalized Sharpe Ratios by \v{C}ern\'{y} \cite{c}. \\
\indent More recently, Cherny and Madan \cite{cm} have built a theoretical framework for these  \emph{acceptability indexes}. Their axiomatic  definition  stems from an analysis of the theoretically and practically desirable properties   a performance criterion should possess.  For  bounded  positions,  they   show a complete characterization of an acceptability index in terms of a parametric family of coherent risk measures   and propose some new law invariant indexes.  The setup in \cite{cm} is static, in the sense that there are only two dates of interest: today and the horizon $T$, and consequently the index is real valued - possibly infinite.  \\
\indent Dynamic   acceptability indexes and their dual representations have been thoroughly  analyzed in Bielecki at alii \cite{bcz}, in a finite $\Omega$, finite set of dates $\T$ framework.     Rosazza Gianin and Sgarra \cite{rs} remove  the requirement of scale invariance  and provide a (static) characterization of these generalized acceptability indexes in terms of a family of \emph{quasi convex} risk measures on  $L^{\infty}$. Then,  they work in   a Brownian context and focus on dynamic acceptability indexes generated by $g$-expectations and give applications to liquidity risk quantification.\\
\indent Our purpose  is to deepen the theoretical analysis of  dynamic performance criteria.  In Section 2,  over a general probability space and stochastic basis, we   define a Conditional Performance Measure (CPM) for positions bounded from below.

The fact that the domain of the CPM includes not only bounded positions but also those  bounded from below leads to reconsider the continuity axioms.
While the current literature  requires an index to be  continuous from above, we ask for continuity from below.

Thus,  we come out with a performance criterion which   encompasses conditional expected utility and  conditional  acceptability indexes as a particular cases, as shown in the Examples in Section 2.1. Then, in Section 3 we provide the representation of a CPM in terms of an induced family of risk measures.  Such family is in fact  \emph{convex}, and not only quasi convex, so we  complete and extend to the conditional case the intuition for the static case in \cite[Proposition 3]{rs}. The results in Section 3  are, in order:
\begin{itemize}
  \item in Section 3.1, we prove a CPM $\bt$ induces a  parametric family of convex risk measures $(\rtz)_z$,  with $z\in (z_d, z_u)$ where $z_d<z_u$  are  the (non-random) essential infimum and supremum of the CPM.  The properties of $(\rtz)_z$ are listed in Proposition \ref{induced-risks};
  \item  in Section 3.2 we define standard families of risk measures, and we show that any such  family generates a CPM;
  \item   $(\rtz)_z$   is the unique standard family which generates $\bt$. This is shown in Proposition \ref{sumup}.
\end{itemize}

In Section 4, the   dynamic case is considered. A  Dynamic Performance Measure (DPM)  is naturally defined as a collection of CPMs and  the time consistency problem is addressed. The notion we choose is very close to the one in \cite[Proposition 6]{rs}, and basically ensures that positions which are considered good tomorrow are already considered good today.   In virtue of the results in Section 3, a time consistency requirement on the DPM must have an equivalent in terms of the time consistency of the induced families of convex risk measures. In fact, Proposition \ref{consistency} shows our notion of time consistency for a DPM is equivalent to  the weak acceptance consistency for the induced risks.  We conclude with some examples and counterexamples on DPMs (Subsection 4.2), including DPM generated by $g$-expectations,   and with a note on the extension of  CPMs and DPMs to dividend  processes (Subsection 4.3).

 \section{ Conditional performance measures and acceptability indexes}
A fixed time  horizon $T>0$ is given  and  the trading dates $ \T$ are a  subset of $[0,T]$, typically either a  finite subset  $\{t_1, t_2,\ldots t_n \},t_n=T,$ or the entire interval. The stochastic basis $(\Omega, (\Ft)_{t\in \T}, P) $ models the possible outcomes, the evolution of the information with time and the probability of the various events.    In the following, $t$ is always assumed to be in $\T$.
For a shorthand, let us denote by $L^{\infty}_t$ the space $L^{\infty}(\Omega,\Ft, P )$  of essentially bounded variables and  let  $\Lbbt$ denote  the lattice of essentially  bounded from below ones, possibly $+\infty$-valued.  Equalities and inequalities are intended to hold  $P$-almost surely. For example, $X>0$ means $P(X>0)=1$, and  we refer to elements of $\Lbbt$ simply as 'bounded from below' variables and to elements in $\L^{\infty}_t$ as 'bounded'.  The space of  simple variables in $\Lbbt $, i.e. those $\mathbb{R}$-valued  and assuming only a finite set of real values, is denoted by $\St $.   The conditional expectations $E[ \cdot \mid \Ft], E^Q[\cdot \mid \Ft] $ where $Q\ll P$  are denoted by $E_t[\cdot], E_t^Q[\cdot]$.     \\
\indent The convention  adopted hereafter  is  that notation is  self-explaining.  (Equivalence classes of) $\FT$-measurable random variables are denoted  by capital Latin letters  $X,Y, \ldots$,  $\Ft$-measurable random variables by the  Greek letters  $ \xi, \eta $, while  the Greeks $\varphi,\psi$ are reserved for elements of $\St$.   Real numbers will be indicated by lower Latin letters $a,b, c,d\ldots$. However, for clarity's sake  there will be  complete  statements like e.g. $c \in \mathbb{R}$ in the definitions and where any misunderstanding is possible.

\begin{definition}\label{perf}
A performance measure, conditional to $\Ft$, is a map $\beta_t: \LbbT \rightarrow \Lbbt $  with the  following six properties:
\begin{enumerate}
\item  quasi-concavity: for any real  $c \in[0,1] $  and for any $X,Y \in \LbbT$
     $$ \bt( c X + (1-c)Y) \geq \min( \bt(X), \bt(Y)) $$

\item if $z_d,z_u$ denote the infimum and supremum of $\bt$,
 $$z_d = \essinf_{X} \bt(X), z_u =   \esssup_{X} \bt (X) $$
  then, $z_d ,  z_u  $ are non-random and  $z_d<z_u$ (possibly infinite);
\item    monotonicity: $\bt$ is  \begin{enumerate}
                        \item  non decreasing, $\bt(X)\geq\bt(Y)$ whenever $X\geq Y$;
                        \item  increasing  over  constant positive shifts:
                        $$  c>0  \Rightarrow  \bt( X+c)>\bt(X )  \text{ on } \{ \bt(X)<z_u \} \cap \{ \bt( X+c)>z_d\} $$
                      \end{enumerate}
\item continuity from below: if $X_n \uparrow X\in \LbbT$, then $\bt(X) =\lim_n \bt(X_n) $;
\item locality:
\begin{equation}
I_B\bt(X)= I_B\bt(X I_B) \  \forall B\in \Ft
\label{loc}
\end{equation}

\item for each $z\in \mathbb{R}$, let $A_{t}^z: =\{ X \in \LbbT \mid \bt(X)\geq z\}$  {denote the} $z$-upper level set of $\bt$.
    If   $z_d < z< z_u$, then   $A^z_t \cap \Lbbt$ is  uniformly bounded from below, i.e. $\essinf A^z_t \cap \Lbbt   \geq x_z$ for some $x_z\in \mathbb{R}$.
\end{enumerate}
For a shorthand, conditional performance measures will be indicated by the acronym CPM.   If the CPM   further satisfies:
\begin{itemize}
  \item[7.]  scale invariance: the map is positively homogeneous of degree $0$,
    $$\bt(c X) = \bt(X), \text{ for all } X    \text{ and reals } c>0$$
  \item[8.] nonnegativity: $\bt\geq 0$,
\end{itemize}
    it is called  Conditional Acceptability Index (CAI).
\end{definition}

Some of the above properties are intuitive: quasi-concavity means diversification is encouraged by the CPM, monotonicity that more is preferred to less.   The requirement of non-randomness on  $z_d, z_u $  and  that of  strict monotonicity on constant positive shifts are technical,  but   quite reasonable and verified in all the examples we will see.

 Continuity from below, Property 4, in turn ensures that profiles which are  'pointwise-close' to a given profile $X$ in the limit cannot underperform $X$.  \\
  \indent To give a  precise mathematical statement for the last assertion,  recall first  the notion of  order convergence in  a vector lattice $L$. A sequence  $(X_n)_n$ is order convergent to $X$, notation $X_n \stackrel{o}{\rightarrow} X$  if there exists a nonnegative sequence $(Y_n)_n$ such that
    \begin{equation} \label{orderconv}
    |X_n - X| \leq Y_n  \text{ and } Y_n \downarrow 0
    \end{equation}
 We refer to Aliprantis and Border \cite[Chapters 8-9]{ab}, for more details on lattices and order convergence, to  Foellmer and Schied
  \cite[Section 4.2]{fs} for applications of these notions to convex risk measures on spaces of bounded variables, and to Biagini and Frittelli \cite{bf-namioka}  when the convex risk measures are defined over general spaces. \\
   \indent The following Lemma is the analog of the equivalence between the Monotone Convergence Theorem  and  Fatou Lemma in Calculus. An immediate corollary is  that for  a CPM $\bt: \LbbT \rightarrow \Lbbt$  continuity from below is equivalent to   order lower semicontinuity
 with respect to  sequences uniformly bounded from below.

\begin{lemma}\label{c.above-lsc}
With the convention $\infty -\infty =0$ and $c+\infty=\infty$,  for a monotone non decreasing map $\pi: \LbbT \rightarrow \Lbbt$  continuity from below is equivalent to   order lower semicontinuity with respect to sequences uniformly bounded from below, i.e. is equivalent to:
$$  X_n\stackrel{o}{\rightarrow}  X  \text{ and } X_n \geq c \text{ for all  } n \Rightarrow \pi(X) \leq \liminf_n \pi(X_n) $$
\end{lemma}
\begin{proof}
 $\LbbT, \Lbbt $ are only lattices and not vector spaces, but any absolute difference $|Y-X|$ is a well defined element of $\LbbT$ under the stated convention. The  only  thing to prove  is  that  continuity from below implies order lsc with respect to sequences bounded from below.
 Assume then $\pi $ is continuous from below and  pick a   sequence $(X_n)_n$, $X_n \geq  c$,  which is
  order  convergent to $X$; then, $X_n \rightarrow X$ a.s.,  so $X\geq c$,  and  there exists a sequence  $Y_n\downarrow 0$  such that $ |X_n-X| \leq Y_n$. So,
  $$ (X-Y_n)\vee c \leq X_n  \text{ and } (X-Y_n)\vee c \uparrow X$$
  whence, passing to the liminf, by monotonicity and continuity from below:
 $$ \pi(X) = \lim_n \pi((X-Y_n)\vee c)  \leq \liminf_n \pi(X_n) $$
\end{proof}

Note that  the monotonicity of $\bt$, $+\infty \in \LbbT$ and continuity from below of $\bt$ imply
 \begin{equation}\label{zeta-u}
   z_u = \bt(+\infty)= \lim_{ x\uparrow +\infty} \bt(x)
 \end{equation}
 Since  $\bt$ is defined on $\LbbT$,  monotonicity also implies
 $$z_d= \lim_{x\downarrow -\infty} \bt(x)$$
Therefore, for  any $z_d \leq z <z_u$, the upper level set $A^z_t $ contains a constant.  This  will be often used in the following, in the form

\begin{equation}\label{six-a}
z_d \leq z <z_u \Rightarrow A^z_t\cap L^{\infty}_t \neq \emptyset
\end{equation}

Property 5, locality, is rather natural.  On a set of states $B$ revealed  at $t$,  the  performance of the future profile  $X$ is uniquely determined by the restriction of $X$ to $B$, all the other states being irrelevant. \\
  \indent  The last property of CPM, property 6,   is  a mild technical assumption,   as it is  verified in a number of interesting cases as shown in the Examples in Section 2.1.
Finally, scale invariance and nonnegativity are  requirements for (static) acceptability indexes as introduced by Cherny and Madan \cite{cm}.  Scale invariance  for a performance means it evaluates the goodness of the whole direction of trade   generated by a profile $X$, namely its ray, rather than $X$ only.  This   may be appropriate  only for a large investor,  for whom quantity does not matter, and when the liquidity risk is negligible.

\begin{remark}\label{contbelow} The motivations supporting  the choice of  continuity from below  for $\bt$, instead of continuity from above  as assumed in \cite{cm,rs,bcz}, are manifold.
\begin{enumerate}
  \item Some   natural performance measures  as conditional expectation/expected utility/certainty equivalent on $\LbbT$ are \emph{not} continuous from above, but are continuous from below.
  \item \emph{Continuity from below solves the value-at-$0$ puzzle for indexes.} In the set up of Cherny and Madan, continuity from above implies that an index unbounded above must be necessarily $+\infty$ in $0$. This  is rather awkward for any index, but especially for those  null on the negative constants or on the whole  negative orthant.     In our context  an index   null on negative constants   in $0$ is naturally $0$ valued. See Lemma  \ref{scaleinvariance} and Remark \ref{expconsist} for a proof of this simple fact, and  Section 2.1  for the  concrete example  of  the  Bernardo and Ledoit's Gain-Loss Ratio   index  and the Best Gain-Loss Ratio index introduced by Biagini and Pinar \cite{bp}.
  \item As shown in Section 3,  continuity from below of $\bt$   implies   the associated risk measures $\rtz$ are also continuous from below. This in turn implies  that the risks, when restricted to   $L^{\infty}$,  have a  special dual representation which \emph{implies} continuity from above  on $L^{\infty}$. This special dual representation, with a maximum over the dual variables,  is essential in any practical application.  See e.g. the examples given by \cite{cm}.
  \item Last, but not least, in  many specific cases there is a natural vector  space (a Banach lattice in fact) over which the map can be defined. This space is $L^1$ for the case of Gain-Loss Ratio, $L^p, p<+\infty$ for Lower Partial Moments analysis,  or, more generally, an Orlicz  heart as in  the Examples \ref{CCE}, \ref{glr} and \ref{rrr}.  If the  map  is e.g. a (quasi concave) ratio of concave monotone increasing  function over a convex monotone  decreasing function,  it will  be norm continuous in the interior of its proper domain by the extended Namioka-Klee Theorem. Since these norm topologies are order continuous, the map will be  globally order continuous, from above \emph{and } below,  on the interior of the proper domain.  See \cite{bf-namioka} for more details on this important point.

\end{enumerate}

\end{remark}

\begin{lemma}\label{scaleinvariance}
Suppose a CPM $\bt$ is scale invariant. Then, for all  $\xi \in \Lbbt$

 \begin{equation}\label{scale}
 \bt(\xi) =  \bt(1) I_{\{ \xi>0\}}  + \bt(0) I_{\{ \xi \leq 0\}}
 \end{equation}
 and consequently $\bt(X)= \bt(0)=z_d$ for all $X\in \LbbT, X\leq 0$ and $ \bt(X) = \bt(1)=z_u$ for all $X$ such that $X \geq \xi >0$ for some $\xi \in \Lbbt$.\\
  \indent In particular, when $\bt$ is unbounded above ($z_u=+\infty$), then  $\bt(X) = +\infty$ on all $X\in \LbbT$ such that $X\geq \xi >0$ for some $\xi \in \Lbbt$.
\end{lemma}
\begin{proof} Suppose first $\xi>0$ and pick $\varphi_n \in \St$,  $\varphi_n \uparrow \xi$. Note $\varphi_n $ cannot be chosen strictly positive in general.  Since $\varphi_n$ is simple, the locality property and scale invariance  imply   $ \bt( \varphi_n ) I_{\{\varphi_n >0 \}} =  \bt(\varphi_n I_{\{\varphi_n >0 \}}) I_{\{\varphi_n >0 \}}= \bt(I_{\{\varphi_n >0 \}}) I_{\{\varphi_n >0 \}}$. By continuity from below  then
  $$\bt(\xi) = \lim _n \bt(\varphi_n ) = \lim_n \bt( \varphi_n ) I_{\{\varphi_n >0 \}}  = \lim_n  \bt(I_{\{\varphi_n >0 \}}) I_{\{\varphi_n >0 \}} = \bt(1)$$
On the other side, locality,  scale invariance and  continuity from below give  $\beta_t(\varphi)=\beta_t(0)$ for all $\varphi \leq 0$ in $S_t$. Let  $\xi \in \Lbbt$, $\xi\leq 0$. Then, $\xi$ is bounded and if $c$ indicates a negative lower bound, by monotonicity $ \bt(c) \leq \bt(\xi) \leq \bt(0)$, whence $\bt(\xi)=0$ and   \eqref{scale} follows by localization.  The conclusions for $X\in \LbbT$ are easy consequences of \eqref{scale} and of the monotonicity property. \\
\indent  The last assertion  follows from \eqref{zeta-u} and \eqref{scale}.
\end{proof}
\begin{remark}[On CAIs] \label{expconsist} Thanks to the above Lemma, a CAI  which is: 1) unbounded above; and 2)  expectation consistent, i.e. $\bt(X) = 0 $
for all $X$ such that $E_t[X]\leq 0$ and $\bt(X)>0$ if $E_t[X]>0$, verifies

$$ \bt(\xi) = +\infty I_{\{ \xi>0\}} + 0 I_{\{ \xi\leq 0\}} \ \text{ on }  \Lbbt$$
and  $\bt(X)= \bt(0)=0$ for all $X\in \LbbT, X\leq 0$ and $ \bt(X) = +\infty $ for all $X$ such that $X \geq \xi >0$ for some $\xi \in \Lbbt$.
\end{remark}

\bigskip
Locality of the CPM ensures that upper level sets are also local, in the sense explained below.
\begin{lemma} \label{pasting}
Let $\beta_t$ be a CPM and suppose $X_1,X_2$ satisfy:
$$  \bt(X_1) \geq z   \text { on } B\in \Ft,  \bt(X_2) \geq z  \text { on } B^c.$$
Then $Y = X_1 I_B + X_2 I_{B^c} \in A^z_t$.
\end{lemma}
\begin{proof}
From locality of $\bt$, $$  I_B \bt (Y ) = I_B \bt (Y I_B)= I_B \bt (X_1 I_B) = I_B \bt(X_1) \geq z I_B$$
Similarly  on $B^c$ we get $ I_{B^c}  \bt (Y) \geq z I_{B^c}  $ and therefore $Y\in A^z_t$.
\end{proof}

\begin{corollary} \label{simple}
 If $z_d\leq z<z_u$ and   $ \bt(X) \geq z   \text { on } B\in \Ft$, then there exists $\widetilde{X}\in \LbbT$, $ \widetilde{X} I_B = XI_B$ such that $\bt(\widetilde{X})\geq z$. The same result holds when inequalities are strict, i.e. if  $ \bt(X) >z$ on $B$, then there is $\widetilde X$ such that $\bt(\widetilde{X})>z$.
\end{corollary}
\begin{proof}
When $z<z_u$, $A^z_t\cap L^{\infty}_t \neq \emptyset$ by  \eqref{six-a}.   Pick $\xi$ is this set, let  $X_1=X$, $X_2=\xi$ and define $\widetilde{X} =Y$ where $Y$ is as in the previous Lemma. For the case of strict inequalities, pick $\xi $ in $A^{z'}_t\cap L^{\infty}_t  $ for some $z_u > z'>z$ and define $\widetilde{X}$ in the same way.
\end{proof}
\subsection{Examples of CPMs and CAIs}

 \begin{example}[Conditional expected (random) utility]\label{ut}
Consider a  function $U: (-\infty, +\infty] \times \Omega \rightarrow (-\infty, +\infty]$ which is $\mathcal{B}((-\infty, +\infty])\otimes \FT$-measurable. We assume it is a stochastic utility function, in the sense that for every fixed $\omega$, $U(\cdot, \omega)$ is a   finite-valued,  concave,  monotone nondecreasing function. Note $U(\cdot, \omega)$ is a fortiori continuous  on $\mathbb{R}$.\\
\indent   In order to define a map via  conditional expectation   extra assumptions must be imposed. In fact $U(X(\cdot),\cdot)$ is $\FT$ measurable when $X$ is, but   conditional expectation may not be either  well-defined or $\Lbbt$-valued on $\LbbT$.  Also,    Properties 2  and 6  of CPMs do not hold for a generic stochastic utility as it can be seen from variations of  the items in the examples list below (e.g. if in item 3 the endowment $W$ is not bounded).   Suppose then:
 \begin{itemize}
   \item[a) ] for all $x \in \mathbb{R}$,  $ U(x, \omega)\in \LbbT$
   \item[b) ] $U(+\infty,\omega)= \sup_x U(x, \omega)$ and $ E_t[U(+\infty, \omega)] \in (-\infty, +\infty]$
   \item[c) ] for any  fixed level $z<z_u$,  $  \{ \xi \in \Lbbt  \mid     E_t[U(\xi, \omega)] \geq z   \} $  contains a bounded variable  and is  bounded from  below.

 \end{itemize}
Define on $\LbbT$:
$$ \bt(X):= E_t[U(X(\omega), \omega)] $$
 Then $\bt$ is a CPM. In fact, it is easy to verify that item a)   implies $\bt $  is well defined and  $\Lbbt$-valued.  This map  is evidently concave and  monotone.  By the monotone convergence Theorem, assumptions a) and  b),   $\bt$ is continuous from below (Property 4 of CPM).  Assumption b) and  monotonicity give:
$$  E_t[ U(+\infty, \omega)] \geq \esssup_X  E_t[U(X(\omega), \omega)]  \geq \esssup_x  E_t[U(x, \omega)]   =  E_t[U(+\infty, \omega)]$$
 which implies $z_u = E_t[ U(+\infty, \omega)] \in (-\infty, +\infty] $, while assumption c) implies $z_d=-\infty$ so that Property 2 of CPMs holds,    Property 5 is automatically satisfied given the structure, as
$$ U(X,\omega) I_B = U(X I_B,\omega) I_B$$
with the convention $ 0 \cdot \infty =0$. Property 6 directly   follows from  assumption c), so
  $\bt$ is a CPM. \\
\indent This set up covers, among others, the next cases.
\begin{enumerate}
  \item Conditional expectation with respect to $Q\ll P$, $Q_{|\Ft}\sim P_{|\Ft}$. Let $ Z_T:=\frac{dQ}{dP}$, call $Z_t =E_t[Z_T]$, pick two versions of these variables and set $U(x,\omega) = x \frac{Z_T}{Z_t}(\omega)$. Assumptions  a), b) and c)  hold   as $E_t[U(\cdot, \omega)] = E_t^Q[\cdot]$ and  is thus the identity on $\Lbbt$.
  \item Conditional expected utility. When $U(x, \omega) =  U(x)$ and $U(\infty)=\sup_x U(x)$, and $\inf_x U(x)= -\infty$ 
we are  in the classic, non random utility case.  Assumption a) and b) are trivially satisfied.
   By concavity $U$ is strictly increasing up to $ \overline{x} = \inf \{ x \mid U(x)=z_u =U(+\infty)\}$.  Then, for any   $z<z_u  $  $U^{-1}$ 
is well-defined and  $E_t[U(\xi)] = U(\xi) \geq z $  iff $ \xi \geq U^{-1}(z)$, so that  $A^z_t \cap \Lbbt $ is bounded from below.
  \item  Conditional expected utility with random  endowment $ W  \in L^{\infty}_t$. Here $U(x, \omega) =   \widetilde{U}(x +W( \omega))$ where $\widetilde{U}$ is a deterministic utility function as in the previous example.  Again assumption a) and b)  are easily verified for $\bt(X) = E_t[ \widetilde{U}(X +W)]$. Assumption  c) is verified as in the previous example.
  \item Conditional expected exponential utility with random risk aversion. Let
  $$ U(x, \omega) = 1- e^{ - \gamma(\omega) x} $$
 with $  \gamma \in (\Lbbt)_+$ , $\gamma \geq c >0$. The restrictive measurability assumption on $\gamma$ is due to  the financial interpretation as risk aversion is known today, while the  uniform lower bound on risk aversion   guarantees assumption c) is satisfied.
\end{enumerate}
\end{example}

 \begin{example}[Conditional Certainty Equivalent, CCE]\label{CCE}
 Suppose we are given   a stochastic dynamic utility:
  $$ U = U(x,t,\omega), t \in \mathcal{T}  $$
   where for fixed $t$  the utility satisfies the same conditions as in the previous example and  the two extra conditions: $U(\cdot, t, \cdot) = U_t$ is $ \mathcal{B}(\mathbb{R})\times \Ft$-measurable and   $x\rightarrow U(x, t,\omega)$ is a.s. strictly increasing.    Fix then $t<T$ and suppose the assumptions  a) and b)
on $E_t[U(\cdot, T, \omega)]$  hold as in the previous example.    Then,
  the CCE at time $t$   is the  unique variable  $C_t(X)$ such that
 $$  U( C_t(X), t,\omega) =  E_t[ U(X,T, \omega)]  $$
 and it maps $ \LbbT$ in $\Lbbt$.  The CCE  is easily recognized to be a CPM.  From the point of view of the preferences,
 it offers  a different  numerical representation on $\LbbT$  of the  preference order induced by $E_t[U(\cdot, T,\omega)]$. Its economic meaning stems from its interpretation as insurance premium,  called mean value principle (see Gerber \cite[Chapter 5.4]{g}).   For the definition of the  CCE on more general spaces, the reader is referred to  Frittelli and Maggis \cite{fm}.  Kupper and Schachermayer \cite{ks} have shown that the  CCE process plays a considerable role in  dynamic performance measures, see   Section 4 for more on this topic.
 \end{example}

\begin{example}[Gain Loss Ratio and Best Gain Loss Ratio]\label{glr}
 Bernardo and Ledoit (\cite{bl}) introduced the static Gain-Loss Ratio as an alternative to the Sharpe Ratio in performance evaluation.  Cherny and Madan \cite[Section 3.2]{cm} have shown that the  static Gain Loss Ratio
 $$  GLR_0 (X) =   \left  \{ \begin{array}{cc}
                            \frac{E[X]}{ E[X^-]} & \text{ on  } E[X]> 0  \\
                             0  &   \text{ otherwise }
                             \end{array} \right.$$
restricted to $L^{\infty}_T$ is an acceptability index in their sense, i.e. continuous from above.  Note however that   the very definition is slightly inconsistent within their setup, as the map is not  continuous from above at $0$, since on positive, arbitrarily small payoffs it takes the value $+\infty$. It is in fact \emph{continuous from below.}
 The conditional version of GLR is, with the convention $ \frac{\xi}{0} = +\infty $ if  $\xi>0$:
 $$  GLR_t (X) =   \left  \{ \begin{array}{cc}
                            \frac{E_t[X]}{ E_t[X^-]} & \text{ on  } E_t[X]> 0  \\
                             0  &   \text{ otherwise }
                             \end{array} \right.
                             $$
 $GLR_t$ is   a CAI.   It is evidently  well defined on $\LbbT$, $\Lbbt$-valued, monotone, nonnegative and scale invariant.  Quasi concavity follows from convexity of the function $x^-$ and linearity of (conditional)  expectation exactly as in \cite[Section 3.2]{cm}.  Also, $\bt(+\infty) = z_u= +\infty$ and $z_d =0$.  Continuity from below is easily verified, as  from monotone convergence $X_n \in \LbbT$,  $X_n \uparrow X$ implies $E_t[X_n] \uparrow E_t[X]$ and dominated convergence gives $E_t[X_n^-] \downarrow E_t[X^-]$. In particular  the \emph{static $GLR_0$} is  continuous from below.
  Property 6  of CPMs  is also satisfied, since  when   $z>0$
 $$A^z_t\cap \Lbbt = \{ \xi \in \Lbbt \mid \bt(\xi) \geq z \} =\{ \xi \in \Lbbt, \xi>0\}   $$
 We only need to check locality. Note that when $B\in \Ft$, $P(B)\neq 0$:
 $$  \bt (XI_B) =   \left  \{ \begin{array}{cc}
                            \frac{E_t[X]}{ E_t[X^-]} & \text{ on  }  \{E_t[XI_B]>0\}=B \cap \{E_t[X ]>0\}   \\
                             0  &   \text{ otherwise }
                             \end{array} \right.
                             $$
    so that $ \bt(XI_B)I_B = \bt(X)I_B$. \\
    \indent  The Gain Loss Ratio analysis  has been recently extended by Biagini and Pinar \cite{bp}. On a general probability space and in the presence of a market,  they show that the  market optimized  GLR (Best GLR, BGLR) is in turn a static acceptability index as a function of the random endowment $W$:
 $$ BGLR_0 (W) = \sup_{R \in \mathcal{R}} \frac{E[(R+W)]}{ E[(R+W)^-]} $$
 in which $ \mathcal{R}$ are the  claims replicable at zero cost. The conditional version of  BLGR
 $$  BGLR_t (W) = \sup_{R \in \mathcal{R}} \frac{E_t[(R+W)]}{ E_t[(R+W)^-]},  $$
  is also a CAI.
\end{example}

 \begin{example}[Reward to Risk Ratio]\label{rrr}
A generalization of the GLR is   the  conditional  reward to risk  ratio,   under the usual convention $ \frac{\xi}{0} = +\infty $ if  $\xi>0$:
 $$  \bt( X) =  \left  \{ \begin{array}{cc}
                                \frac{E_t[U(X)]}{\sigma_t(X)} & \text{ on  } \{ E_t[U(X)] >0 \}\\
                                0  &   \text{ otherwise }   \\
                             \end{array} \right.  $$
 in which
 \begin{itemize}
   \item  $U$ is a
 concave
utility function, possibly stochastic as in Example \ref{ut},  verifying $U(+\infty)>0$
   \item $\sigma_t:\LbbT \rightarrow \Lbbt$ is a  nonnegative, monotone non increasing,
convex,  local and continuous from below  map, with  $\sigma_t (+\infty) \in (0,\infty)$.
 \end{itemize}
   It can be proved, along the same lines of   the GLR case, that  $\bt$ is a CPM, with $z_u = \frac{U(+\infty)}{\sigma_t(+\infty)}$ and $z_d=0$.    If $U$  and $\sigma$ are  positively homogeneous, then $\bt$ is a CAI. \\
\indent An example is the  ratio of the conditional expectation  to the conditional $p$-th lower partial moment. This has been studied by Leitner \cite{lei-lpm} in the static case, with  $\sigma_0 = E[(X^-)^p]^{\frac{1}{p}}$ and  fixed $p>1$. Its conditional version is thus the CAI
$$  \bt( X) =  \left  \{ \begin{array}{cc}
                                \frac{E_t[X]}{( E_t[(X^-)^p])^{\frac{1}{p}}} & \text{ on  } \{ E_t[X] >0 \}\\
                                0  &   \text{ otherwise }   \\
                             \end{array} \right.  $$
 Another class of CAIs is obtained from conditional coherent risk measures,
see  Bion-Nadal \cite{jbn0} or Detlefsen and Scandolo \cite{ds}  for information  on conditional risks.  Fix a   conditional coherent risk measure $\rho_t: L^{\infty}_T \rightarrow  L^{\infty}_t$ and continuous from below.  Then,  extend it  to $\LbbT$  using the continuity from below  property and finally   set $ \sigma_t : = \rho_t \vee 0$.  To fix the ideas, $\sigma_t$  could be the truncation of the conditional Average Value at Risk at level $ \lambda $, $\lambda \in L^{\infty}_t, 0< \lambda<1$, that is
  $$ \rho_t (X) = \esssup_{\{ Q\mid   Q_{| \Ft} = P_{| \Ft}, \frac{dQ}{dP}\leq \lambda^{-1}\}} E_t^Q[-X]$$
  The resulting   CPM is
 $$  \bt( X) =  \left  \{ \begin{array}{cc}
                                \frac{E_t[U(X)]}{ \rho_t \vee 0} & \text{ on  } \{ E_t[U(X)] >0 \}\\
                                0  &   \text{ otherwise }   \\
                             \end{array} \right.  $$
  which is a CAI when $U$ is positively homogeneous.
  Note that when $t=0$ and $U (x) =x$, this class boils down to the  RAROC class in \cite[Section 3.4]{cm}. The only difference is  we do not set $\bt =+\infty$ where $\rho_t\leq 0$. This choice  in fact ignores the effect of   reward on acceptable claims. Here,  $\bt=+\infty$ on $\rho_t\leq 0, E_t[U(X)]>0$ and null on $ \rho_t\leq 0, E_t[U(X)]\leq 0 $.
 \end{example}

\begin{remark}
In the two examples above, GLR and the more general Reward to Risk Ratio, the index is set to zero when the reward is non positive. This is not surprising, since in the end any performance measure must be optimized. Then of course the optimizer(s) (or at least the maximizing sequence) will have positive reward. For an effective analysis of generalized Sharpe Ratios, the interested reader may consult \v{C}ern\'{y} \cite{c}.
\end{remark}

 \section{Parametric families of conditional convex risk measures and conditional performance measures: a one-to-one relation}
As anticipated in the Introduction,  this section generalizes   and completes   the findings of \cite{bcz, cm, rs}  for  acceptability indexes.

 \subsection{From a CPM to the  induced  family of risks }
\begin{definition}
Let $z_d< z<z_u $ and  consider $A^z_t=\{ X\mid \bt(X)\geq z\}$.  Define
$$ \rtz (X) : =    \essinf \{ \xi \in \Lbbt \mid \bt(X+\xi)\geq z\}=\essinf \{ \xi \in \Lbbt \mid  X+ \xi \in A^z_t \}
  $$
  Note $\rtz$ is in the form of \textbf{capital requirement} with initial set of acceptable position $A^z_t$. The reader unfamiliar with this notion is referred to  \cite[Chapter 4.1]{fs}.
\label{defrho}
 \end{definition}

 The next Lemma shows  some   properties of $\rtz$ and auxiliary results.  Note that if  $M^{z}_t(X): =\{\xi \in \Lbbt \mid \beta_t(X+\xi) \geq z\}$, then $ \rtz(X) =\essinf M^z_t(X)$.
 \begin{lemma} \label{downward}
 Let $z_d < z<z_u$ and  fix $X \in \LbbT$. Then
 \begin{enumerate}
   \item $M^z_t(X)$ contains a bounded variable, it is a sublattice of $\Lbbt$ and is monotone in $X$, namely  $$X\leq Y \Rightarrow M^z_t(X)\subseteq M^z_t(Y)$$
  \item $\rtz(X)\in \Lbat$;
  \item there exists  a  sequence $  \xi_n \downarrow \rtz(X)$, $\xi_n\in M^z_t(X)$;
\item  $\rtz$ is local: $ I_B  \rtz (X) = I_B \,\rtz(X I_B)$ for all $B\in \Ft$;
\item  $\rtz$ is translation invariant on $L^{\infty}_t$: $\rtz(X+ \eta) =  \rtz(X)- \eta$ for all  $\eta\in L^{\infty}_t $
\item   $  \bt(X)> z     \text{ on } B \in \Ft  $ if and only if    $\rtz(X)< 0  \text{ on } B$.   As a consequence,
      \begin{equation} \label{eqlar}
       \bt(X)\leq z     \text{ on } B \in \Ft  \Leftrightarrow    \rtz(X)\geq 0  \text{ on } B
       \end{equation}
 \end{enumerate}
 \end{lemma}
\begin{proof}
\begin{enumerate}

    \item  $M^z_t(X)$ is monotone  by monotonicity of $\bt$.  In fact, $X\geq x$ for some real $x$. From  \eqref{six-a},  as $z_d<z<z_u$  there exists   $\xi$ in  $L^{\infty}_t$ such that  $\bt(\xi)\geq z$;  from monotonicity, $ \bt(X-x +\xi)\geq z$ so that $\xi-x\in M^z_t(X)$ and $\xi-x$ is bounded.  To show closure with respect to lattice properties,  let $\xi, \eta$ in $M^{z}_t(X)$ and  $A=\{\xi \leq \eta\}$.
                    From the locality property:
                    $$\beta_t(X+\inf(\xi,\eta))=1_A\beta_t(1_A(X+\xi))+1_{A^c}\beta_t(1_{A^c}(X+\eta)\geq z$$
                    This proves that $M^{z}_t(X)$ is closed for $\wedge$. The closure for $\vee$ is trivial from the monotonicity of $\beta_t$. Thus it is a sublattice of $\Lbbt$.
    \item  The proof of the previous item shows $\rtz(X) \leq \xi-x$,
{
for some $\xi$ in $L^{\infty}_t$.}
     \item  This follows from  the lattice property of $M^z_t(X)$ and \cite[Chapter VI.1]{nev}. 
 \item  Locality. Fix $B\in \Ft$.   Suppose  $ \bt(XI_B + \xi')\geq z $. In particular, the inequality holds on $B$. 
 In the same way, for any $\xi \in M^z_t(X)$,   $ \bt(X+\xi)\geq z$, and  the inequality holds in particular on $B^c$. Set $X_1 = (X I_B+\xi' ), X_2= (X+\xi)$.  Thanks to Lemma \ref{pasting}, $ X_1 I_B + X_2 I_{B^c} = X + (\xi' I_B + \xi I_{B^c}) \in A^z_t$.  Therefore, for any $\xi' \in M^z_t(X I_B)$, $(\xi' I_B + \xi I_{B^c}) \in M^z_t(X )$ and
      $$ \rtz(X) I_B \leq  (\xi' I_B + \xi I_{B^c}) I_B = \xi' I_B   $$
      so, taking the essinf on $\xi'$ , $ \rtz(X) I_B \leq \rtz(XI_B) I_B$.  The converse inequality is analogous.

\item  Translation invariance by bounded variables   is an immediate consequence of the definition of $\rtz$.
 \item    We prove the first statement.  \begin{itemize}

        \item[$\Rightarrow$] Suppose $\bt(X)>z$ on $B$. Continuity from below gives $\bt(X - \frac{1}{n}) \uparrow \bt(X)$. Set $B_n = \{\omega \in B,  \bt(X - \frac{1}{n})(\omega)>z\}$. Corollary \ref{simple},  locality and translation invariance by bounded variables  of $\rtz$  give
            $$ \rtz(X)I_{B_n} \leq - \frac{1}{n} I_{B_n}$$
            Thus  $ \rtz(X) <0$ on  $B_n$    and, passing to the limit, $\rtz(X)<0$ on $B$.

        \item[$\Leftarrow$]  Suppose $\rtz(X)< 0$ on $B$. For all $n \in \mathbb{N}^*$, let $B_n=\{  \rtz(X) < - \frac{1}{n} \} \cap B$. Then $B_n\uparrow B$ and from translation invariance, $\rtz(X-\frac{1}{n})<0$ on $B_n$.
        From  \eqref{six-a} there is a (bounded) $\xi$ such that $\beta_t(\xi) \geq z$. From translation invariance of  $\rtz$,  $\rho^z_t(\xi+1) \leq -1<0$.  Set $X' =(X-\frac{1}{n}) I_{B_n}+(\xi+1)I_{{B_n}^c}$. Then, locality of the risk map implies $\rtz(X')<0$. Select an approximating sequence   $\xi_k \downarrow \rtz(X')$ as in item 3 above.  By locality and monotonicity,
        $$  \bt(X')  I_{\{ \xi_k < 0\}}\stackrel{loc}{=} \bt(X' I_{\{ \xi_k < 0\}} )  I_{\{ \xi_k < 0\}} \geq  \bt( (X' +\xi_k) I_{\{ \xi_k < 0\}})  I_{\{ \xi_k < 0\}}  \geq z I_{\{ \xi_k < 0\}} $$
        and,  sending $k$ to infinity, we get $\bt(X')\geq z$.
         Property 3-b) of CPMs (strict monotonicity)    ensures  $\bt(X'+\frac{1}{n}) >  \bt(X') $ on $   \{ \bt(X' )<z_u \} \cap \{ \bt( X'+\frac{1}{n})>z_d\} $.  The locality property of $\bt$ and $z>z_d$ gives then  $\bt(X) = \bt(X' + \frac{1}{n}) > z $ on $B_n \cap   \{ \bt(X' )<z_u \} $. This result is satisfied  for all $n$, and since on $  \{ \bt(X' )= z_u \} $ by monotonicity  $\bt(X)=z_u >z$,  then  $\bt(X)> z $ on $B$.

        \end{itemize}
        To show \eqref{eqlar}, suppose $\rtz(X)\geq 0$ on $B$. Let $A=\{  \bt(X) > z\}\cap B \in \Ft$. On $A$,   $\rtz(X)< 0$. It follows that $P(A)=0$. Thus   $\bt(X)\leq  z $ on $B$. The proof of the other implication is similar.
\end{enumerate}
\end{proof}

\begin{proposition}\label{induced-risks}
\begin{itemize}
\item[i)] For  fixed  $z_u > z>z_d$,  consider  the   map $\rtz$. Then, $ \rtz : \LbbT \rightarrow \Lbat$  and is
 \begin{enumerate}
  \item $L^{\infty}_t$-valued if restricted to $\L^{\infty}_T$,

  \item   convex;

  \item monotone non increasing, and  it can be calculated as:
      \begin{equation}\label{bounded} \rtz(X)= \essinf \{ \xi  \in L^{\infty}_t \mid X+ \xi \in A_t^z\}
      \end{equation}

  \item local, $ I_B  \rtz (X) = I_B \,\rtz(X I_B)$ for all $B\in \Ft$

  \item translation invariant on $\Lbbt$:  $\rtz(X+\xi ) = \rtz(X)-\xi$,  for all  $\xi \in \Lbbt$;

   \item continuous from below:    $$ X_n \uparrow X \Rightarrow  \rtz(X_n)\downarrow \rtz(X) $$

\end{enumerate}
    and is therefore a conditional convex risk measure on $\LbbT$, continuous from below.

    \begin{itemize}
      \item[7.] If  $\bt$ is    scale invariant (in particular, when $\bt$ is a CAI),   $\rtz$ is a conditional coherent risk measure on $\LbbT$.
    \end{itemize}

 \item[ii)]  Moreover,  for any fixed  $X \in \LbbT$, the process $ \rtz(X)_{z_d<z<z_u}$ is  a non decreasing  and continuous function of $z \in (z_d,z_u)$. \\
In case  $z_d=-\infty$,
   \begin{equation}\label{limit}
        \ \lim_{z\downarrow -\infty} \esssup \rtz(0) =-\infty
   \end{equation}

\end{itemize}

\end{proposition}

\begin{proof}
\begin{itemize}
\item[i)]
By Lemma \ref{downward}  item 2,  $\rtz$ is $\Lbat$-valued.
\begin{enumerate}

\item Bounded on bounded positions. Suppose $X\in \L^{\infty}_T$. From  Lemma \ref{downward},
$ M^z_t(\|X\|_{\infty} )\supseteq M^z_t(X)$.  So, $ \rtz(\|X\|_{\infty} ) \leq \rtz(X)$.
  Now, $ \|X\|_{\infty}  +  M^z_t(\|X\|_{\infty} ) \subseteq A_t^z\cap \Lbbt$.  Since $z>z_d$, Property 6 of CPMs ensures that  $A_t^z\cap \Lbbt$ has a uniform  lower bound $x_z$. Therefore $$ x_z -  \|X\|_{\infty}  \leq    \essinf M^z_t(\|X\|_{\infty}) = \rtz(\|X\|_{\infty} ) \leq \rtz(X)$$

\item  Convexity. Fix $X,Y$ and take a   convex combination $ Z= c X + (1-c)Y$ with $c \in [0,1]$. For any $\xi  \in M^z_t(X), \psi \in M^z_t(Y),$ then   quasi concavity of $\bt$  implies
  $$ \bt (Z + (c \xi +(1-c)\psi)\,) =  \bt( c (X+\xi) + (1-c) (Y+\psi)\,)\geq z $$
  that is $ \rtz(Z) \leq c \xi +(1-c)\psi $. Taking   the essential infimum over $\xi$ and $\psi$,  $ \rtz(Z)\leq c \rtz(X) +(1-c) \rtz(Y)$.

\item Monotonicity  (in the opposite direction) follows easily from    monotonicity  of $\bt$.  Since  $M_t^z(X)$ contains a bounded variable $\xi^*$ and is a sublattice of $\Lbbt$ by item 1, Lemma   \ref{downward},  $M':=\{ \xi \wedge \xi^* \mid \xi \in M_t^z(X)\} \subseteq M^z_t(X) \cap L^{\infty}_t$  and $\rtz(X) = \essinf M_t^z(X) = \essinf  M' $.

\item Locality  is proved in  Lemma \ref{downward}, item 4.

\item  A combination of  localization and translation invariance over $L^{\infty}_t$ gives translation invariance for all  $\xi \in \Lbbt$. In fact, suppose first $ \xi <+\infty$.  Then,
       $$ (\rtz   (X ) - \xi) I_{\{  \xi \leq n \}} = (\rtz   (X ) - \xi \wedge n) I_{\{  \xi \leq n \}} =  \rtz\left (X + \xi \wedge n  \right ) I_{\{  \xi \leq n \}} =  \rtz(X + \xi )I_{\{  \xi \leq n \}}$$  for all $n$ and the result follows by letting $n\uparrow \infty$. If $B : =\{ \xi = +\infty\}$ has positive probability, then  $$\rtz( X + \xi)I_B = \rtz(+\infty)I_B = (-\infty) I_B $$
       where the second equality (convention $0\cdot \infty =0$) follows from $M^z_t (+\infty )=\Lbbt$.  Therefore for any $\xi \in \Lbbt$
        $$ \rtz(X + \xi) = \rtz(X + \xi)I_{B^c} + \rtz(X + \xi)I_{B} = (\rtz(X) - \xi)I_{B^c} + (- \infty) I_B =\rtz(X) -\xi $$
        since   $\rtz(X)$ is bounded above and thus on $B$ $ \rtz(X) -\xi  =-\infty$.

\item Continuity from below.  Let $X_n \uparrow X$ and suppose  the set  $C = \{ \rtz(X) < \lim_n \rtz(X_n) \}$ has positive probability.  For $q \in \mathbb{Q}$, let $C_q=\{ \rtz(X)< q \leq\lim_n \rtz(X_n) \}$. Since  $C=\bigcup_{q\in \mathbb{Q}}C_q$,  $C_q$
 has positive probability  for some $q$.
 By translation invariance and Lemma \ref{downward} item 6, such chain has two implications:   $\bt(X +q )>z$ on $C_q$,  and  $\bt(X_n +q ) \leq z$ on $C_q$. From the continuity from below of $\beta_t$ the last inequality implies that $\bt(X +q ) \leq z$ on $C_q$,  a contradiction.
       \end{enumerate}

  \begin{itemize}
   \item[7.] when $\bt$ is scale invariant    for any $k>0$
         $$  \rtz( k X) = \essinf \{ \xi \mid  \bt(k X+\xi)\geq z \}=  k \, \essinf \{ \xi' \mid  \bt(X+\xi')\geq z \} =k\rtz( X) $$
         where the second equality holds thanks to scale invariance of $\bt$.
 \end{itemize}

\item[ii)]  The process $ \rtz(X)_{z_d<z<z_u}$ is non decreasing in $z$ as a simple consequence of  monotonicity of $\bt$.
    In particular for any fixed $X$,  any fixed $z^* \in (z_d, z_u)$, $\rho^z_t(X)$ has a left limit and a right limit in $z^*$. Assume that $\lim_{z\uparrow z^*}\rho^z_t(X) \neq  \lim_{z \downarrow z^*}\rho^z_t(X)$.
    Then there are $p<q$ in $\mathbb{R}$ such that
    $A_{p,q}=\{\lim_{z \uparrow z^*}\rho^z_t(X) <p<q \leq  \lim_{z\downarrow z^*}\rho^z_t(X)\}$ has positive probability.  By Lemma \ref{downward}, item 6, and translation invariance, for all $u<z^*$, $\beta_t(X+p)>u$ and for all $v>z^*$,  $\beta_t(X+q) \leq v$ on $A_{p,q}$. Passing to the limit it follows from the monotonicity of $\beta_t$ that  $\beta_t(X+p)=\beta_t(X+q)=z^*$ on $A_{p,q}$. This is in contradiction with the  strict monotonicity  of CPM on constant positive shifts,  property   3.(b) in Definition \ref{perf}.  \\
    To conclude the proof,  the $z$-monotonicity of  $\rtz(0)$  implies the existence of the limit
    $ l:= \lim_{z \downarrow -\infty} \esssup \rtz(0)$.   For any $n \in \mathbb{N}$ there is $z_n$ such that $ \bt(-n) \geq z_n$ and by monotonicity of $\bt$ we   may assume $z_n$ is non increasing.  So, $\rho_t^{z_n}(0)\leq -n$ and
    $$\lim_{z \downarrow -\infty}  \esssup \rtz(0) \leq \lim_{n\rightarrow +\infty}  \esssup \rho_t^{z_n}(0) =-\infty  $$
\end{itemize}    
\end{proof}

\begin{remark}\label{strict-inclusion}
The inclusion $A^z_t \subseteq \{ \rtz\leq 0\}$ is  always satisfied by definition of $\rtz$, but  may be strict due to lack of  order upper semicontinuity.  So in general $A^z_t$ is only a subset of the acceptance set of $\rtz$. For example,  pick any  static ($t=0$) CAI unbounded above and null over negative variables, e.g. the GLR  from  Example \ref{glr}.   Then, for any $z>0$
$$ \rho_0^z(0) = \essinf \{ c \mid  \beta_0 (c) \geq z\} = 0 $$
{
Indeed  $ \beta_0(\frac{1}{n} ) =+\infty $ from Corollary \ref{scaleinvariance}} and thus $\frac{1}{n}\in  A^z_0 $ for every $n$. But $ \beta_0 (0) =0 $ so that  $0\in \{ \rho^z_0 \leq 0 \} \setminus A^z_0$.
\end{remark}
\begin{remark}\label{usc-risk-map}
Continuity from below for the risk $\rtz$, which is monotone non increasing, is equivalent to order upper semicontinuity wrt sequences uniformly bounded from below:
$$ X_n \geq c, \ |X_n -X| \leq Y_n \downarrow 0 \Rightarrow  \rtz(X)\geq \limsup_n \rtz(X_n)   $$
The proof, mutatis mutandis, is identical to that of Lemma \ref{c.above-lsc}.
\end{remark}

\begin{corollary}\label{rep}

For all $z\in (z_d, z_u)$, the map $ \rho^z_t$ admits the following dual representation on $L^{\infty}_T$
\begin{equation}
\rho^z_t(X) = \mathrm{ess}\max_{Q\in{\cal Q} } \, ( E^Q_t[-X] -\alpha_t(Q) ),
\label{eqrep1}
\end{equation}
where: ${\cal Q}$ is the set of probability measures   on $(\Omega,{\cal F}_T)$  absolutely continuous wrt $P$ and  such that  $Q|_{\Ft}=P|_{\Ft}$; and $\alpha_t$ is the penalty function
 $$\alpha_{t}(Q)=\esssup_{Z \in L^{\infty}({\cal F}_T)}(E^Q_t[-Z] -  \rho^z_t(Z))$$
 Therefore, the restriction of the  risk measure to $L^{\infty}_T$ is   \emph{ continuous from above}. Moreover, for all $X$  in $L^{bb}_T$,
\begin{equation}
\rho^z_t(X) = \lim_{n \rightarrow \infty} ( E^{Q_n}_t[-(X \wedge n)] -\alpha_t(Q_n) )
\label{eqrep2}
\end{equation}
for a suitable sequence  $Q_n \in{\cal Q}$.
\end{corollary}
\begin{proof}
The first assertion is the well known representation result for conditional risk measures continuous from below, see Bion-Nadal \cite{jbn0,BN2008}.  It is well known that the representation \eqref{eqrep1}  implies  \emph{ continuity from above} of the restriction of the  risk measure to $L^{\infty}_T$. Finally,
since $\rho^z_t(X)= \lim_n \rtz(X\wedge n) $   \eqref{eqrep2} follows directly from \eqref{eqrep1} and continuity from below.
 \end{proof}

\begin{definition}\label{rtz}

The family $(\rtz)_{z_d<z<z_u}$  is called the family of risk measures induced by the CPM $\bt$.
\end{definition}

\begin{example}
 The  family induced by  conditional exponential utility with random risk aversion, $\bt(X)= E_t[ 1- e^{-\lambda_t X}] $ as in Example \ref{ut}, item 4, is perhaps the simplest example of all.  It  is a family of   conditional entropic risk measures  where the parameter $z$  gives only a translation effect.   In fact,   fix the risk aversion $\lambda_t \in L^{\infty}_t,  \lambda_t$ positive and bounded away from zero. Then, $z_d = -\infty, z_u= 1 $ and, for any $z\in (z_d, z_u)$,     $\rtz (X) =\essinf \{ \xi \mid  E_t[ 1- e^{-\lambda_t (X+\xi)}]\geq z \}$. Explicitly:
 $$ \rtz(X) = \frac{\ln  E_t[e^{-\lambda_t X} ] }{\lambda_t} - \frac{\ln(1-z)}{\lambda_t} $$
 Note $\rtz(0)=  - \frac{\ln(1-z)}{\lambda_t}$, so these risk measures are not normalized. For the level $z=0$ we   recover the normalized  entropic risk measure.
\end{example}

\bigskip

Knowing the whole  family of risks induced by a  CPM  amounts to the knowledge of the  CPM itself.

\begin{proposition}\label{reconstruction}
The CPM $\bt$ can be reconstructed from the induced family of risks $ (\rtz)_{z_d<z<z_u}$. In fact,  suppose first  $z_d \in \mathbb{R}$ and set   $B_X := \{ \bt (X)=z_d\}$.  Then

 \begin{enumerate}
   \item  $B_X= \cap_{z \in (z_d,z_u)} \{ \rtz(X)\geq 0\}$
   \item  with the convention   $0\cdot \infty =0$, $\bt(X)=  z_d  I_{B_X} +   I_{B_X^c} \xi$, where $$ \xi:=\esssup \{  \varphi \in \St \mid z_d< \varphi<z_u,  \varphi = \sum_{i=1}^n  z_i I_{B_i  }, \rtzi(X) < 0   \text{ on } B_i\cap B_X^c \text{   for all } i=1\ldots n \}$$
 \end{enumerate}
In case $z_d =-\infty$, we get the simplified relation $$ \bt(X) = \esssup \{  \varphi \in \St \mid  \varphi<z_u,  \varphi = \sum_{i=1}^n  z_i I_{B_i  }, \rtzi(X) < 0   \text{ on } B_i  \text{   for all } i=1\ldots n \}. $$
\end{proposition}

\begin{proof}
Item 1 is a consequence of \eqref{eqlar}. To prove item 2, note that  by localization we can assume $\bt(X)>z_d$ everywhere.  The following proof will then hold also in case $z_d=-\infty$.    Note
\begin{eqnarray*}
   \bt(X) &=& \esssup \{ \varphi \in \St \mid \varphi <   \bt(X) \}\\
 & = & \esssup \{   \varphi \in \St \mid z_d<\varphi <z_u,  \varphi = \sum_{i=1}^n  z_i I_{B_i},  \bt(X) > z_i  \text{ on } B_i \text{  for all } i=1\ldots n \}\\
    & =  & \esssup \{   \varphi \in \St \mid  z_d< \varphi <z_u,  \varphi = \sum_{i=1}^n  z_i I_{B_i},  \rtzi(X) < 0   \text{ on } B_i \text{  for all } i=1\ldots n \}\\
    & =&\bt(X)
\end{eqnarray*}
in which  the  first two equalities give  simply the approximation of a bounded from below claim by bounded variables, while   the last two  hold thanks to  Lemma \ref{downward}, item 6 first part.
\end{proof}

 \subsection{From a parametric family of risks to the induced CPM}
\label{sectcc}
A one-parameter family of conditional risk measures $(\stz)_{z}$ which satisfies certain regularity assumptions induces a CPM. The idea is to take the relations in items 1, 2 of Proposition \ref{reconstruction} as  definitions  of a map $\bt$ and  to show  it is indeed a CPM.

\begin{definition}[Standard family of risk measures]
   Fix a non empty open interval $I = (z_d,z_u)\subseteq \mathbb{R}$, possibly unbounded. Let $(\stz)_{z\in I}$ be  a parametric family of conditional convex risk measures continuous from below on $\LbbT$, non decreasing and  continuous in the parameter $z$, i.e. a collection of maps satisfying the conditions:

\begin{itemize}
  \item[a)] For fixed $z\in I$,  $\stz: \LbbT \rightarrow \Lbat$ is:
           \begin{enumerate}
             \item   $L^{\infty}_t$-valued when restricted to $L^{\infty}_T$;
             \item convex, i.e. $\forall c \in [0,1]$, $\forall X,Y \in \LbbT$, $\stz(cX+(1-c)Y) \leq c\stz(X)+(1-c)\stz(Y)$;
             \item monotone non increasing
             \item translation invariant on $\Lbbt$: $\stz(X+ \xi) = \stz(X)-\xi$ whenever $\xi \in \Lbbt$;
             \item local:  $\stz(X) I_B = \stz(XI_B) I_B$ for all $B\in \Ft$;
             \item continuous from below: $X_n \uparrow X \Rightarrow \stz(X_n)\downarrow \stz(X)$.

           \end{enumerate}
  \item[b)] For fixed $X \in \LbbT$,    the paths of the process  $ ( \stz(X))_z$  are almost surely  non decreasing and   continuous on $I$.
\item[c)]   If $z_d=-\infty$,  then $\lim_{z\rightarrow -\infty } \esssup \stz(0)  =  -\infty$.

      \end{itemize}

Then, the family of risk measures  $(\stz)_{z\in I}$  is called standard.
\end{definition}

\begin{proposition}  \label{risk2beta}
Let $(\stz)_{z\in I}$ be a standard family of  conditional risk  measures.    Then,  with the convention $0 \cdot \infty =0$ define on $\LbbT$ the map $\bt$ as
\begin{equation}\label{risk2beta-eq}
 \bt(X):=z_d I_{B_X} + I_{B_X^c} \esssup\{\varphi \in \St \mid   z_d< \varphi <z_u,   \varphi =\sum_1^n z_i I_{B_i}, \stzi(X)<0 \text{ on } B_i  \cap B_X^c \}
  \end{equation}
in which   $B_X:= \cap_{z\in I} \{\stz(X) \geq 0 \}$.
The map $\bt$ is a CPM,  called the CPM  {\bf generated  by} $(\stz)_{z\in I}$,  and the thresholds $z_d, z_u$ are the essential infimum and supremum of $\bt$.    If  $  (\stz)_{z \in I}$ is a family of conditional coherent (i.e. sublinear) risk measures and if $z_d=0$,  $\bt$ is a CAI.

\end{proposition}

\begin{proof}
The proof is split in a number of steps.
\begin{enumerate}
 \item   $\bt: \LbbT \rightarrow \Lbbt$.

 Fix $X\in \LbbT$.  When $z_d$ is finite, by definition $\{\bt(X) = z_d\} \equiv  B_X$.  So, if  $z_d$ is finite we are done.    In case  $z_d =-\infty$,   by property c) of standard families, there exists  $z^* \in \mathbb{R}$ such that $ \sigma^{z^*}_t(0)< -\|X^-\|_{\infty}$ and  by monotonicity and  translation invariance  $ \sigma^{z^*}_t(X)\leq  \sigma^{z^*}_t(  - \|X^-\|_{\infty})<0$, from which $\bt(X) \geq z^*$.
 \item Locality.  \begin{eqnarray*}
          \bt(XI_B)I_B &=&     z_d I_{B_X\cap B} +  \\
                  && I_{B_X^c \cap B}\,  \esssup\{    \varphi =\sum_1^n z_i I_{B_i}, z_i \in I,  \stzi(X I_B)<0 \text{ on } B_i  \cap B_X^c \}  \\
             &=&   I_B ( z_d I_{B_X}\, +  I_{B_X^c}  \esssup\{   \varphi =\sum_1^n z_i I_{B_i}, z_i \in I, \stzi(X  )<0 \text{ on } (B_i\cap B)  \cap B_X^c \}
        \end{eqnarray*}
   where the second equality follows from locality of $\stzi$.       When $\{ B_1, \ldots B_n\}$  varies among all finite $\Ft$-measurable partitions of $\Omega$,  $ \{ B_1\cap B, B_1\cap B^c, \ldots,  B_n\cap B, B_n \cap B^c\} $ still varies among all such partitions and therefore we get  that the last term displayed above equals $ I_B \bt(X) $.
  \item If $z_u>z>z_d$, then   $ \bt(X) > z$ on $B$ if and only if  $\stz(X)<0$ on $B$.    The arrow $(\Rightarrow)$ is a consequence of the definition \eqref{risk2beta-eq}.  Conversely,  suppose  $\stz(X)<0$ on $B$.  By the $z$-(right) continuity  and monotonicity of the process $\stz(X)_z$ we have
      $$ \{\stz(X)< 0\} = \cup_{n\in \mathbb{N}^*, z+\frac{1}{n}<z_u} \{\sigma^{z+ \frac{1}{n}}_{t}(X)< 0\} $$
      so that  $B = \cup_{n\in \mathbb{N}^*, z+\frac{1}{n}<z_u} B_n$, where  $B_n:=B \cap  \{\sigma^{z+ \frac{1}{n}}_{t}(X)< 0\}$.  But on $B_n$, by definition of $\bt$ and by its locality, $\bt(X)\geq z+ \frac{1}{n}>z$.

\item Quasi concavity. By locality, it is enough to show that,  for any $z_u>z>z_d$, if  $\bt(X)>z $ and $\bt(Y)> z$, then $ \bt(cX+(1-c)Y) >z$. In virtue of item 3 above, this is equivalent to $\stz(X)<0, \stz(Y)<0$  imply $ \stz(cX+(1-c)Y) <0$, which holds by convexity of $\stz$.
  \item Thresholds. Let $z \in I$.  Since $\stz(0) \in L^{\infty}_t$,   $\stz(0)\leq b_z$ for some constant. Therefore,  for all $n$ $\stz(n) \leq b_z-n$, and thus from decreasing monotonicity, $\stz(+\infty) =-\infty$ for all $z$.  This implies  $\bt(+\infty) =z_u $.  To show $z_d$ is the lower threshold, set
  $$ \lim_n \bt(-n) : =    \eta $$
  The limit exists as $\bt$  is clearly monotone non decreasing. The set $\{\eta>z_d \}$ has zero probability. In fact,
  $$ \{\eta>z_d \} =\cup_{z_u> q>z_d , q \in \mathbb{Q}} \{\eta>q \}  $$
On  every  $\{\eta>q \}$   $\bt(-n) >q$ for all $n$ and this    implies $\sigma^q_t(0)<-n$ for all $n$, which leads to a contradiction unless  $P( \eta>q)=0$.

 \item Non decreasing monotonicity is clear, so we just show strict monotonicity over constant positive shifts.   Let $c>0$. Assume that $\bt(X) <z_u, \beta_t(X+c)>z_d$.   Set $B = \{ \bt(X+c) =\bt(X) \}$ and let us prove $P(B)=0$.    On $B$, for any $ z_u>z>z_d $ such that $\bt(X+c)>z$, also $\bt(X)>z$.  Reformulating thanks to   item 3 above, for any such   $z$,   $\sigma^z_t(X+c)<0$ on $B$ implies $\sigma^z_t(X)<0$ on $B$. Thus  on $B$  the image of the path of $\stz(X)_z$ does not contain  $[0,c)$. Due to the continuity of the map $z \rightarrow \sigma^z_t(X)$ it follows that  on $B$, $\sigma^z_t(X)<0$ for all $z$ or that
 $\sigma^z_t(X) \geq c$ for all $z \in (z_d,z_u)$. The first condition means $\bt(X)=z_u$ and the second $\bt(X+c)=z_d$, a contradiction unless $P(B)=0$.

 \item Continuity from below.
Let $X_n\uparrow X$. Then,  from   monotonicity of $\bt$, $\bt(X)\geq \lim_n \bt(X_n)$.   We obtain an equality as soon as  if for some $z_u>z>z_d$   $\bt(X)> z$ on  $B$, then  $\lim_n \bt(X_n)> z$ on $B$.  By locality, since $\bt(XI_B) I_B = \bt(X)I_B$, by eventually replacing $X$  with $XI_B$  and $X_n$ with $X_n I_B$, we can assume $B=\Omega$. By  item 3 above,   $\stz(X)<0$.  From continuity from below of the risk measure $\stz$, $ \stz(X_n) \downarrow \stz(X)$.  If $ C_n: =  \{ \stz(X_n )<0  \}$, then $C_n \uparrow \Omega$ and by definition of $\bt$
 $$ \lim_n \bt(X_n) \geq \bt(X_n) > z \text{ on } C_n  $$
 whence $ \lim_n \bt(X_n) >z$.
\item Property 6 of CPM.  If  $  z_d < z < z_u$ and  if $z-\epsilon >z_d$, then   $\{ \bt \geq z \} \subseteq  \{ \bt > z-\epsilon \} = \{ \sigma^{z- \epsilon}_t <0 \}$ by monotonicity and item 3.  So,
           $$ \xi \in \Lbbt, \xi \in A^z_t \Rightarrow   \sigma^{z- \epsilon}_t (\xi) < 0 \Rightarrow   \sigma^{z- \epsilon}_t (0) < \xi $$
           and since $  \sigma^{z- \epsilon}_t (0) \in L^{\infty}_t$ we conclude $A^z_t \cap \Lbbt $ is uniformly bounded from below.

\item Properties 7 and 8 for CAIs are easily verified  in case  $  (\stz)_{z\in I}$ is a family of conditional coherent risk measures and  $z_d=0$.
 \end{enumerate}
\end{proof}

 \subsection{A one to one relation between  standard families of risks and CPMs}
\begin{proposition}\label{sumup}
 We tie  together the results shown so far:
\begin{enumerate}
  \item a CPM $\bt$ induces a standard family of convex risk measures $(\rtz)_{z_d<z<z_u}$; if $\bt$ is a CAI, the risks are coherent;
  \item a standard family of convex risk measures generates a CPM; if in addition the risks are coherent, they induce a CAI;
  \item   $(\rtz)_{z_d<z<z_u}$ is the  unique  standard family of convex risk measures generating  the CPM $\bt$, in the sense that  any other standard family $(\stz)_{z_d<z<z_u}$ generating $\bt$ satisfies
       $$\rtz(X) = \stz(X) $$

\end{enumerate}

\end{proposition}
\begin{proof}
\begin{enumerate}
  \item The family $(\rtz)_z$ is standard from Proposition \ref{induced-risks}.
  \item This is Proposition \ref{risk2beta}.
  \item  If $(\rtz)_z$ is the family induced by $\bt$, then the family in turn generates  a CPM which by Proposition \ref{reconstruction} is exactly $\bt$.  If, for any other standard family $(\stz)_{z_d<z<z_u}$,  we show that the equality $\{\stz\geq 0 \}= \{\rtz\geq 0 \}$ holds,  then the uniqueness of  $(\rtz)_z$  will follow from translation invariance.
  Fix then $X\in \LbbT$.   From \eqref{eqlar} and from the proof Proposition \ref{risk2beta}, item 3,
  $$  \rtz(X) \geq 0 \text{ iff } \bt(X) \leq z \text{ iff } \stz(X)\geq 0 $$
\end{enumerate}

\end{proof}

\begin{remark}
  If $\Ft$ is trivial and if we are in the scale invariant case,  the above Proposition  becomes a static result which can be rephrased as follows:

   \begin{flushleft}
   \emph{ A map $\beta: \LbbT \rightarrow [0, +\infty]$ is an acceptability index if and only if there exists a unique  standard family of coherent risk measures $(\sigma^z)_{z \in (0,z_u)}$      such that
     $$   \beta(X)= \sup \{ z \mid  \sigma^z(X) <  0   \} $$
    with the convention $\sup \emptyset =0$.
 The  standard family is    $(\rho^z)_{z \in (0,z_u)}$,  with  $\rho^z(X) =\essinf \{ c \mid  \beta(X+c)\geq z\}$.}

   \end{flushleft}

This is the ($\LbbT$-formulation of the) result in \cite{cm}, Theorem 1 and Proposition 2 eqq.  (4) and (5).  Here the authors  obtain a very close  expression for $\beta$ as $\beta(X)= \sup \{ z \mid  \sigma^z(X) \leq  0   \} $, with  non strict  inequality, and they get a  maximal standard family, possibly non unique.  Such difference  is due to   the different  order continuity assumption of the map, from below and not from above, and  to our assumption of strict monotonicity over constant positive shifts.
  Our result also  generalizes to the conditional case   some of the findings of Drapeau and Kupper for static quasi-convex risk measures \cite[Theorem 1.7]{dk}.

\end{remark}

\section{Dynamic Performance Measures}

\subsection{Definition   and  time consistency}

\begin{definition}\label{DPM}
A Dynamic Performance Measure (DPM) is a family    $\beta = (\beta_t)_{t\in \T }$   of CPMs with the properties:
\begin{enumerate}
  \item $\bt: \LbbT \rightarrow \Lbbt$
  \item  the (essential) supremum and infimum $ z_u, z_d$  do not depend on $t \in \T$. \footnote{From a mathematical viewpoint, we could relax the requirement to $ z_u =z_u(t)$ is non increasing, $z_d=z_d(t)$ is non decreasing with no substantial changes.}
\end{enumerate}
 In case the  CPMs are in particular CAIs,  we call the DPM  a Dynamic Acceptability Index (DAI).
\end{definition}
The consistency across time of a DPM certainly is a key issue. The notion we choose ensures that if a future performance measure  $\bt$ assesses a  value to  a payoff $X$ uniformly \emph{ strictly}  higher than a certain level, the today's performance measure $\bs$   preserves such level.
\begin{definition}
A DPM (or DAI)  $\beta$ is   time consistent  if for all $z_d <z<z_u$, $X\in \LbbT$,  $s, t \in \T, s<t$
\begin{equation}\label{tc}
  \bt(X) > z   \Rightarrow  \bs(X) > z ,
\end{equation}

\end{definition}

\begin{remark}\label{rqwt}
  Locality of  each  component of the DPM $\beta$     implies   that the time consistency definition can be given in the equivalent form $$ X \in \LbbT, B \in \Fs,  \bt(X) >  z \text{ on }  B  \Rightarrow \beta_s(X) > z \text{ on } B  $$

The simple proof  relies on Corollary \ref{simple}.  A  very close notion of time consistency has been given in  Rosazza and Sgarra \cite{rs} with non strict inequalities in the above displayed relation,
for the case of DAIs continuous from above originating from $g$-expectations defined on bounded variables.
\end{remark}

\begin{proposition} \label{consistency}
  For a DPM $\beta$
 the following are equivalent:
 \begin{enumerate}
   \item  $\beta$ is time consistent;
   \item   for all $z_d<z<z_u$, $ s<t$, $ \{\rtz < 0  \} \subseteq \{ \rsz <0 \} $;
\item for all $z_d<z<z_u$ $ s<t$, $ \{\rtz \leq 0  \} \subseteq \{ \rsz \leq 0 \} $.
 \end{enumerate}
\end{proposition}
\begin{proof}
\bigskip
\begin{itemize}
\item[] The equivalence $1) \Leftrightarrow 2) $  is an easy consequence of Lemma \ref{downward}, item 6.

\item[$2) \Rightarrow 3)$]
Let $s < t$ and suppose  $\rho^z_t(X) \leq 0$. By translation invariance,   $\rho^z_t(X+ \frac{1}{n})<0$. From property 2),  $\rho^z_s(X+ \frac{1}{n})<0$. By translation invariance, passing to the limit we get  $\rho^z_s(X) \leq 0$.
\item[ $3) \Rightarrow 2)$] Conversely, suppose 3) holds and fix $X$ such that  $\rho^z_t(X)<0$. Let $B_n=\{ \rho^z_t(X)<-\frac{1}{n}\}$. By translation invariance  and locality $\rho^z_t((X-\frac{1}{n})I_{B_n}+XI_{B_n^c})\leq 0$, from which   $\rho^z_s((X-\frac{1}{n})I_{B_n}+XI_{B_n^c})\leq 0$. By localization and translation invariance $\rho^z_s(X) < -\frac{1}{n}$ on $B_n$ and   passing to the limit we conclude $\rho^z_s(X)<0$.
\end {itemize}
\end{proof}

\begin{remark}
The equivalence of the properties 2) and 3) in the above proposition is very general. It applies as soon as the translation invariance property and the local property are satisfied. In particular it is satisfied for dynamic risk measures on $L^p$ spaces, $1 \leq p \leq \infty$.
\end{remark}
\begin{remark}
 In the specific case of normalized risk measures Property 3 of Proposition \ref{consistency} is known under the name of 'weak acceptance consistency', see Acciaio and Penner \cite{ap} for  a detailed overview of time consistency definitions for risks and detailed references.
\end{remark}
\begin{corollary}
Let $z_d<z<z_u$. For $Q \ll P$, let
$$\alpha^{z}_{t}(Q):=\esssup_{Z \in L^{\infty}({\cal F}_T)}(E^Q_t[-Z] -  \rho^z_t(Z))$$
Then, if the  DPM $\beta$   is time consistent,  for any fixed $z$ we have
\begin{equation}
 \text{ for all }   Q \ll P, s<t,\;\; E^Q_s[\alpha^{z}_{t}(Q) ] \leq \alpha^{z}_{_s}(Q)
\label{eqpenalty}
\end{equation}
\end{corollary}
\begin{proof}
Acciaio and Penner  \cite{ap} prove  that for    dynamic risk measures on $L^{\infty}_T $,  weak time consistency  is equivalent to (\ref{eqpenalty}). They assume the risks are normalized, but the same proof holds for the non normalized case. Therefore from Proposition \ref{consistency},  equation (\ref{eqpenalty}) is a necessary condition for time consistency of the DPM.
\end{proof}

On $L_T^{\infty}$, the acceptance sets satisfy the following property:
\begin{lemma}
Fix $z_d<z<z_u$ and $0\leq  t <T$.   Then  $\{ \rtz\leq 0\}\cap L^{\infty}_T$  is weak* closed   and
$$\overline{ \{ \rtz< 0\}\cap L^{\infty}_T}^* =  \{ \rtz\leq 0\}\cap L^{\infty}_T $$
\end{lemma}
\begin{proof}
{The risk measure $\rtz$ restricted to $L^{\infty}_T$ takes values in $L^{\infty}_t$.}
As a consequence,  $\{ \rtz\leq 0\}\cap L^{\infty}_T$ is exactly the acceptance set of the risk $\rtz$ restricted to $L^{\infty}_T$. Such set is then weak* closed because, by a standard result on conditional risk measures on $L^{\infty}_T$ \cite{jbn0}, continuity from below  implies continuity from above,  which in turn is equivalent to the weak* closure of the acceptance set. \\
\indent  To prove  $  \{ \rtz\leq 0\}\cap L^{\infty}_T = \overline{\{ \rtz< 0\}\cap L^{\infty}_T}^*  $, it  is enough to show  the inclusion $\subseteq$. So, pick  a bounded $Y , \rtz(Y) \leq 0$  and simply consider $ Y_n = Y + \frac{1}{n}$. Then,  $\rtz(Y_n) \leq -\frac{1}{n}$,  so  $Y_n \in \{ \rtz< 0\}$, and  $Y_n \rightarrow Y $ in the uniform norm topology  and in particular in the weak* one.
\end{proof}

\subsection{Time consistent DPMs: examples and counterexamples}
\begin{example} The dynamic version of Example \ref{ut}.  Suppose we are given a stochastic dynamic utility as described at the beginning of Example \ref{CCE}. Then,  $\beta$ defined as
$$ \bt(X) = E_t[U(X,t, \omega)]$$
is a DPM.  Thus, it is time consistent if
$$ E_t[U(X,t, \omega)]> z \Rightarrow E_s[ U(X,s, \omega)]>z  $$
for any choice of dates $s<t$.   Then
\begin{itemize}
  \item if $U=U(x)$, $\beta$ is obviously time consistent
  \item in the exponential case, $ U(x,t,\omega) = 1- e^{-\gamma_t x }$, $\beta$ is time consistent if  the risk aversion process is pathwise non decreasing, i.e.  $ t \rightarrow \gamma_t(\omega)$ is non increasing.
\end{itemize}
\end{example}
\begin{example}[Dynamic Certainty Equivalent]  Suppose the utility $U=U(x)$ is finite on $\mathbb{R}$  and strictly increasing.
  Then, the dynamic certainty equivalent process $C=(C_t)_t$:
  $$  C_t(X) = U^{-1}( E_t[U(X)])$$
 easily seen to be  a time consistent DPM. It is also straightforward to prove that $C_s$ is normalized, i.e. $C_s(z)=z \;\forall z \in \mathbb{R}$, and satisfies the  \emph{strong} time consistency  property: $C_s(C_t(X)) = C_s(X) $. Thus it is straightforward to check that the DPM $C$ is time consistent according to our definition. This example is the only possible one  satisfying the  law invariance, the normalization and the   strong time consistency properties, as shown by   Kupper and Schachermayer
 \cite[Theorem 1.4]{ks}.
\end{example}
\begin{example}[Dynamic Gain-Loss Ratio]  The DAI defined by
$$ GLR_t (X) =   \left  \{ \begin{array}{cc}
                            \frac{E_t[X]}{ E_t[X^-]} & \text{ on  } E_t[X]> 0  \\
                             0  &   \text{ otherwise }
                             \end{array} \right. $$
 is time consistent. If $ z>0$ is fixed, then $ GLR_t (X)> z$  is equivalent to $ E_t[ X^+ - (1+z) X^-]>0$, so that,  when $s<t$,  $E_s[  X^+ - (1+z) X^-]>0$, i.e. $GLR_s(X)>z$. This dynamic index has been used by Bielecki et alii \cite{bciz} to compute price intervals for  path dependent options.
\end{example}
\begin{example}
As pointed out also by Bielecki et alii, \cite[Example 6.5]{bcz}, in general the DCAIs originating from reward to  risk ratio are not time consistent. The reason is that  $E_s[\sigma_t(X)]$ is not in general greater than  $\sigma_s(X)$. For instance, the reward to $p$-th lower partial moment, $1<p<+\infty$, which is but a slight variation of $GLR$:
$$
\left  \{ \begin{array}{cc}
 \frac{E_t[X]}{( E_t[(X^-)^p])^{\frac{1}{p}}}& \text{ on  } E_t[X]> 0  \\
                             0  &   \text{ otherwise }
                             \end{array} \right. $$
is not time consistent as e.g. simple numerical examples in a two step binomial tree  show.

\end{example}
\begin{example} [DAI from g-expectations]
In  \cite{rs}, the notion of dynamic acceptability index  has been introduced in
 the particular case of a Brownian filtration from a family of $g$-expectations $({\cal E}_{g_{\gamma}})_{\gamma}$ defined on $L^2_T$.
 More precisely, fix a Brownian motion $B$. For every $\gamma \in \mathbb{R}_+ $  consider  a driver $g_\gamma$   and   the associated BSDE:
\begin{equation}
-dY_t=g_{\gamma}(t,Y_t,Z_t)dt-Z_tdB_t,\; \forall t \in [0,T],\;\text{and}\; Y_T=X
\label{eqBSDE}
\end{equation}
with   terminal condition $X\in L^2_T$. As in  \cite{rs}, assume $g_{\gamma}(s,\omega,z)$ is  convex and uniformly 
Lipschitz in $z$, and  $g_{\gamma}(s,\omega,0)=0$. The  (unique) solution is a couple of processes $(Y,Z)$ and  $ Y_t$ is denoted by 
$ {\cal E}_{g_{\gamma}}(X|{\cal F}_t)$. Let  $ \rho^{\gamma}_t(X)={\cal E}_{g_{\gamma}}(-X|{\cal F}_t)$. 
The family $ (\rho^{\gamma}_t)_t$ is then a family of  normalized dynamic risk measures continuous from above 
satisfying the strong time consistency property: for all $0 \leq s \leq t \leq T$, $\rho^{\gamma}_s(-\rho^{\gamma}_t(X))=\rho^{\gamma}_s(X)$.
 Moreover it follows from Lemma 36.3 of Peng \cite{Peng}  that $\rho^{\gamma}_t$ defined above is continuous from below on $ L^2_T$. 
Thanks to normalization, translation invariance and monotonicity, $\rho^{\gamma}_t$ is bounded on bounded variables.  
Therefore, any risk measure  $\rho^{\gamma}_t$ can   be uniquely extended  to a  normalized risk measure on $\LbbT$, 
denoted in the same way, which is continuous from below and  is $\Lbat$-valued.  \\
 \indent Suppose now that for a given $t$,  $g_{\gamma}(t,\omega,z)$ is a  continuous function of 
$\gamma$, uniformly in $(\omega,z)$, and is non decreasing. Then, $(\rho^{\gamma}_t)_\gamma$
  is a standard family of conditional risk measures. Call  $\beta_t(X)$ the generated CAI as given by formula \eqref{risk2beta-eq}.  
It is easy to check that $\beta =(\beta_t)_t$  is a DAI.  It satisfies  Property 3  of Proposition \ref{consistency} thus it is time consistent.
\end{example}
  
\begin{example}[DPM from general time consistent dynamic risk measures]
Let $(\Omega,{\cal F},{\cal F}_t,P)$ be a filtered probability space. For all $z_d<z<z_u$, let $\rho^z_{s,t}$, be a  
time consistent dynamic risk measure on $L^{\infty}$ i.e for all
 $0 \leq r \leq s \leq t \leq T$, $\rho^z_{s,t}: L^{\infty}_t \rightarrow  L^{\infty}_s$ and  $\rho^z_{r,t}(X)=\rho^{z}_{r,s}(-\rho^{z}_{s,t}(X))$. 
  Assume that for all $0 \leq s \leq t \leq T$, $\rho^{z}_{s,T}$ is continuous from below  and that $\rho^z_{s,t}(0) \leq 0$.
 Assume also that $z \rightarrow  \rho^z_{s,T}(X)$ is non decreasing and continuous  for all $X$. 
The family $\rho^z_{s}=\rho^z_{s,T}$ is then a standard family of risks. The corresponding DPM is time consistent.\\ 
\indent General examples of  time consistent dynamic risk measures with possible jumps may be obtained   from a stable set
 ${\cal Q}$ of probability measures all equivalent and penalties $\alpha_{s,t}(Q)$ defined for $Q \in {\cal Q}$
 satisfying the local and cocycle properties (see Bion-Nadal \cite{BN2008} and \cite{bn1}). For all $X \in L^{\infty}_t$,
$\rho^z_{st}(X)=\esssup_{Q \in {\cal Q}} (E^Q_s[-X]-\alpha^z_{s,t}(Q))$. 
In particular when ${\cal Q}$ is weakly relatively compact and the penalty is a lower semi-continuous function of $Q$ (which is the case for the minimal penalty), $\rho^z_{st}$ is continuous from below. Assuming furthermore that $z \rightarrow \alpha^z_{s,t}(Q)$ is  non increasing and continuous (uniformly in $Q$),  the family $\rho^z_{s}=\rho^z_{s,T}$ is a standard family of risks. The families of time consistent dynamic risk measures constructed in \cite{BN12-7} 
 in a Markovian setting can easily be used to construct such a standard family of risks.\\
\indent This construction generalizes that of DAIs from g-expectations.  
\end{example}

\subsection{Extension of  performance measures to dividend processes}
Bielecki et alii  introduced in \cite{bcz} a notion of dynamic acceptability index for dividend processes on a finite probability space  with finite dates $\T$.   Our goal here is to show that  there is a natural correspondence between CPMs on random variables and CPMs on processes, and   between DPMs on random variables and DPMs on dividend  processes.

\begin{definition}
A process $D =(D_t)_{t\in \T}$  is called a (discounted) dividend process if
\begin{enumerate}
\item  $D_t  \in \Lbbt,$ for all $t\in \T$
\item  in case $\T$ is infinite,  there is a discrete subset $\T_D$ of $\T$ such that $D_t=0$ if $t \notin \T_D$
\item  $D_t \geq 0$ except for a finite number of $t \in \T_D$.
\end{enumerate}
Denote by ${\cal D}$ the lattice of dividend processes.
\label{defdiv}
\end{definition}

Fix now $t\in \T$.  A CPM $ \bt:\LbbT \rightarrow \Lbbt$  induces a   CPM on dividend processes,   $\hat{\bt}: \cal{D} \rightarrow \Lbbt$, by setting
 $$\hat \bt(D): =\bt(\sum_{r \geq t}D_r))$$

Indeed,  $\hat \beta_t$  is well defined, as $ \sum_{r\geq t}D_r$ is bounded from below from item 3 in Definition \ref{defdiv}, and   $\hat \bt (D) \in \Lbbt$.  The definition of $\hat \beta_t$ implies
 \begin{equation}\label{bhat}
 \hat \beta_t(D) = \hat \beta_t(  \left (\sum_{r \geq t}D_r \right ) I_{\{ T\}})
 \end{equation}
where the notation $  \xi  I_{\{ r \}} $ with $\xi \in L^{bb}_r$ here and in the following indicates the dividend process with a unique cashflow of $\xi$ at time $r$.  So, it is straightforward  to check that  $ \hat \bt$   satisfies the following properties:
\begin{enumerate}

\item independence of the past and locality:
 for $D, D'\in \cal{D}$, if  for some $B\in \Ft$  $D_r I_B=D'_r I_B$ for all $r \geq t$ then $I_B\hat \beta_t(D)=I_B\hat \beta_t(D')$
\item    $\essinf_D \hat \beta_t(D) =z_d, \esssup_D \hat \beta_t(D) =z_u $
\item  non decreasing monotonicity:
if $D_r \geq D'_r$ for all $r \geq t$, then  $\hat \beta_t(D) \geq \hat  \beta_t(D')$;
\item   strict monotonicity over constant positive shifts, i.e.
$$ c>0, r\geq t \Rightarrow \hat \beta_t(D+cI_{\{r^*\}})> \hat \beta_t(D) \text{ on } \{ \hat \beta_t(D + cI_{\{r^*\}})<z_u \} \cap \{ \hat \beta_t(D+c )>z_d\} $$
\item quasi concavity:
 for any    $c \in[0,1] $  and for all $D,D' \in {\cal D}$
     $$ \hat \beta_t( c D + (1-c)D') \geq \min( \hat \beta_t(D), \hat \beta_t(D')) $$
\item translation invariance:  for all $D \in {\cal D}$ and $\xi \in \Lbbt$,  measurable bounded from below, for all $r \geq t$
  $$\hat \beta_t(D + \xi I_{\{t\}})=\hat \beta_t (D + \xi I_{\{r\}})$$
\item scale invariance, if $\bt$ is a CAI.
\end{enumerate}

 A  map $\alpha_t$  on $\cal{D}$ with the  above properties  will be called a CPM for processes (a CAI for processes if scale invariant). In turn any $\alpha_t$    induces a  CPM  $\bt^*$ on $\LbbT$  by the formula
$$ \bt^*(X) =\alpha_t( X I_{\{T\}}) $$
i.e. by restricting $\alpha_t$ to dividend processes  with a single cashflow at the terminal date $T$.  The relation
$$ ({\hat \bt})^* = \bt,  $$
 is clear, while in general $  \widehat{\alpha_t^* } \neq \alpha_t$ unless $\alpha_t$ verifies \eqref{bhat}.  \\
 \indent  A DPM for dividend processes is then a collection $\alpha = (\alpha_t)_{t\in\T}$ of CPMs with  $\alpha_t(+\infty) =z_u, \alpha_t(-\infty):= \lim_n \alpha_t(-n) = z_d$  for all $t$.
 The time consistency property for these DPMs is formulated as the analog of time consistency for DPM on variables:

\begin{definition}
The DPM for processes  $\alpha$ is time consistent  if
$$\alpha_t(D)> z \Rightarrow \alpha_s(D) >z  \text{ for all } D \in \cal{D}, z_d < z<z_u$$
\end{definition}
We conclude by noting that a DPM $\beta$ is time consistent if and only if the induced DPM for processes $\hat{\beta}$ is time consistent. The simple proof is left to the reader.

\begin{remark}
In the same way, one can show  there is  a  correspondence between   conditional/dynamic  risk measures on dividend processes and conditional/dynamic risk measures on random variables.
\end{remark}
\end{document}